\documentclass[11pt]{article}

\usepackage{amsmath,amsthm,amssymb} %
\usepackage[affil-it]{authblk} %
\usepackage[english]{babel} %
\usepackage{caption} %
\usepackage{color} %
\usepackage{algorithmic} %
\usepackage{algorithm} %
\usepackage[utf8]{inputenc} %
\usepackage{csquotes} %
\usepackage{dsfont} %
\usepackage{enumitem} %
\usepackage{flushend} %
\usepackage[T1]{fontenc} %
\usepackage{mathpazo} 
\usepackage{framed} %
\usepackage[margin=1in]{geometry} %
\usepackage{graphicx} %
\usepackage{hyphenat} %
\usepackage[breaklinks]{hyperref} %
\usepackage{mathabx} %
\usepackage{mathtools} %
\usepackage{microtype} %
\usepackage[utf8]{inputenc} %
\usepackage{soul} %
\usepackage{subcaption} %
\usepackage{subdepth} %
\usepackage{suffix} 
\usepackage{tikz} %
\usepackage{xspace} %
\usepackage{thmtools}
\usepackage{thm-restate}

\usepackage{interval}
\intervalconfig{soft open fences}
\usepackage{cleveref}

\hypersetup{colorlinks=true, linkcolor=blue, citecolor=magenta} %
\definecolor{shadecolor}{rgb}{.95,.95,.95} %

\setul{1ex}{.5pt} %
\setlist[enumerate]{nolistsep,itemsep=3pt,topsep=3pt} %

\definecolor{White}{rgb}{1,1,1} %
\definecolor{Black}{rgb}{0,0,0} %
\definecolor{LightGray}{rgb}{.8,.8,.8} %

\usepackage{mathtools}

\newtheorem{theorem}{Theorem} %
\newtheorem{lemma}[theorem]{Lemma} %
\newtheorem{proposition}[theorem]{Proposition} %
\newtheorem{corollary}[theorem]{Corollary} %
\newtheorem{definition}[theorem]{Definition} %

\DeclarePairedDelimiter\rbra{\lparen}{\rparen}
\DeclarePairedDelimiter\sbra{\lbrack}{\rbrack}

\DeclarePairedDelimiter\Abs{\lVert}{\rVert}

\newcommand{\microspace}{\mspace{.5mu}} %

\newcommand{\norm}[1]{\left\lVert #1 \right\rVert}
\newcommand{\norms}[1]{\lVert #1 \rVert}

\newcommand{\ceil}[1]{\left\lceil #1 \right\rceil}
\newcommand{\ceils}[1]{\lceil #1 \rceil}
\newcommand{\abss}[1]{\left\lvert #1 \right\rvert}
\newcommand{\abs}[1]{\lvert #1 \rvert}
\newcommand{\paren}[1]{\left( #1 \right)}
\newcommand{\parens}[1]{( #1 )}
\newcommand{\sqb}[1]{\left[ #1 \right]}

\newcommand{\set}[1]{\left\{ #1 \right\}}
\newcommand{\sets}[1]{\{ #1 \}}
\newcommand{\ket}[1]{\ensuremath{\lvert\microspace #1
    \microspace\rangle}} %

\newcommand{\E}{\mathbf{E}}

\newcommand{\defeq}{\stackrel{\smash{\text{\tiny\rm def}}}{=}} %

\newcommand{\CC}{\mathbb{C}} %
\newcommand{\RR}{\mathbb{R}} %
\renewcommand{\Im}{\operatorname{Im}} %
\newcommand{\diag}{\operatorname{diag}} %
\newcommand{\supp}{\operatorname{supp}} %
\newcommand{\KL}{\operatorname{KL}} %
\newcommand{\TV}{\operatorname{TV}} %
\newcommand{\mix}{\operatorname{mix}} %
\newcommand{\trace}{\operatorname{Tr}} %

\renewcommand{\tilde}[1]{\widetilde{#1}} 
\renewcommand{\bar}[1]{\overline{#1}} 

\def\cD{\mathcal{D}} 
\def\cW{\mathcal{W}} 
\def\cO{\mathcal{O}} 
\def\cR{\mathcal{R}} 
\def\cT{\mathcal{T}} 

\newcommand{\pb}{\mathop{\mathbb{P}}\displaylimits} 

\newcommand{\QRST}{\mathsf{QRST}} %
\newcommand{\RST}{\mathsf{RST}} %
\newcommand{\SubgraphConstruct}{\mathsf{SubgraphConstruct}} %
\newcommand{\QResistance}{\mathsf{QResistance}} %
\newcommand{\Label}{\mathsf{Label}} %
\newcommand{\QStoreTree}{\mathsf{QStoreTree}} %
\newcommand{\QMaxProductTree}{\mathsf{QMaxProductTree}} %
\newcommand{\QIsoSample}{\mathsf{QIsotropicSample}} %

\begin{document}
{\sloppy


\title{\Large\bf Quantum Speedups for Sampling Random Spanning Trees}
\renewcommand*{\Affilfont}{\small\itshape} 
\author[1]{Simon Apers}
\author[2,3]{Minbo Gao}
\author[4]{Zhengfeng Ji\thanks{Correspondence author. Authors are listed in alphabetical order.}}
\author[2,3]{Chenghua Liu}
\affil[1]{
Université Paris Cité, CNRS, IRIF, Paris, France
}
\affil[2]{
  Institute of Software, Chinese Academy of Sciences, Beijing, China
}
\affil[3]{
  University of Chinese Academy of Sciences, Beijing, China
}
\affil[4]{
Department of Computer Science and Technology, Tsinghua University, Beijing, China
}

\date{\today}

\maketitle

\begin{abstract}
  We present a quantum algorithm for sampling random spanning trees from a
  weighted graph in $\widetilde{O}(\sqrt{mn})$ time, where $n$ and $m$ denote
  the number of vertices and edges, respectively.
  Our algorithm has sublinear runtime for dense graphs and achieves a quantum
  speedup over the best-known classical algorithm, which runs in
  $\widetilde{O}(m)$ time.
  The approach carefully combines, on one hand, a classical method based on
  ``large-step'' random walks for reduced mixing time and, on the other hand,
  quantum algorithmic techniques, including quantum graph sparsification and a
  sampling-without-replacement variant of Hamoudi's multiple-state preparation.
  We also establish a matching lower bound, proving the optimality of our
  algorithm up to polylogarithmic factors.
  These results highlight the potential of quantum computing in accelerating
  fundamental graph sampling problems.
\end{abstract}

\section{Introduction}
 Random spanning trees are among the oldest and most extensively studied probabilistic structures in graph theory, with their origins tracing back to Kirchhoff’s matrix-tree theorem from the 1840s~\cite{kirchhoff1847ueber}. This foundational result established a profound connection between spanning tree distributions and matrix determinants. Beyond graph theory, random spanning trees are deeply intertwined with probability theory~\cite{lyons2017probability} and exemplify strongly Rayleigh distributions~\cite{borcea2009negative}. Intriguingly, sampling random spanning trees has revealed unexpected and remarkable connections to various areas of theoretical computer science. Notably, they have been instrumental in breakthroughs that overcame long-standing approximation barriers for both the symmetric~\cite{gharan2011randomized} and asymmetric versions of the Traveling Salesman Problem~\cite{asadpour2010log}. Goyal, Rademacher and Vempala~\cite{goyal2009expanders} demonstrated their utility in constructing cut sparsifiers.
 Beyond their significance in theoretical computer science, random spanning trees also play a crucial role in machine learning and statistics. 
 They have been applied to a wide range of tasks, including graph prediction~\cite{cesa2013random}, network activation detection~\cite{sharpnack2013detecting}, spatial clustering~\cite{teixeira2019bayesian}, online recommendation systems~\cite{herbster2021gang}, and causal structure learning~\cite{duan2023low,richter2023improved}.

Given its significant theoretical importance and broad applications, the problem of sampling random spanning trees has attracted substantial research interest, leading to a series of advancements in performance and algorithmic techniques~\cite{guenoche1983random, broder1989generating, aldous1990random, kulkarni1990generating, wilson1996generating, colbourn1996two, kelner2009faster, madry2014fast, harvey2016generating, durfee2017sampling, durfee2017determinant, schild2018almost, anari2021log}. These contributions can be broadly categorized into three main approaches:

\begin{itemize}
	\item Determinant calculation approaches~\cite{guenoche1983random, kulkarni1990generating, colbourn1996two}: These methods build on Kirchhoff’s matrix-tree theorem, which states that the total number of spanning trees in a weighted graph is equal to any cofactor of its graph Laplacian. Leveraging this theorem, one can calculate determinants to randomly select an integer between 1 and the total number of spanning trees and efficiently map this integer to a unique tree. This approach was first employed by Guenoche~\cite{guenoche1983random} and Kulkarni~\cite{kulkarni1990generating} to develop an $O(mn^3)$ time algorithm,  where $m$ and $n$ denote the number of edges and vertices in the graph, respectively. Later, Colbourn, Myrvold, and Neufeld~\cite{colbourn1996two} improved it to an $O(n^\omega)$ time algorithm, where $\omega \approx 2.37$ is the matrix multiplication exponent.

	\item Approaches based on effective resistances~\cite{harvey2016generating, durfee2017sampling, durfee2017determinant}: 
	The core insights behind there approaches is that the marginal probability of an edge being in a random spanning tree is exactly equal to the product of the weight and the effective resistance of the edge. Harvey and Xu~\cite{harvey2016generating} proposed a deterministic $O(n^\omega)$ time algorithm that uses conditional effective resistances to decide whether each edge belongs to the tree, iteratively contracting included edges and deleting excluded ones. Building on this, Durfee, Kyng, Peebles, Rao and Sachdeva~\cite{durfee2017sampling} leveraged ``Schur complements'' to  efficiently compute conditional effective resistances, resulting in a faster algorithm with a runtime of  $\widetilde O(n ^{4/3}m^{1/2}+n ^2)$.  In a separate work, Durfee, Peebles, Peng and Rao~\cite{durfee2017determinant} introduced determinant-preserving sparsification, which enables the sampling of random spanning trees in $\widetilde{O}(n^2 \epsilon^{-2})$ time, from a distribution with a total variation distance of $\epsilon$ from the true uniform distribution.

	\item Random walk based approaches~\cite{broder1989generating, aldous1990random, wilson1996generating, kelner2009faster, madry2014fast, schild2018almost, anari2021log}: First introduced independently by Broder~\cite{broder1989generating} and Aldous~\cite{aldous1990random}, it was shown that  a random spanning tree can be sampled by performing a simple random walk on the graph until all nodes are visited, while retaining only the first incoming edge for each vertex. For unweighted graphs, this results in an $O(mn)$ algorithm. The subsequent works~\cite{wilson1996generating, kelner2009faster, madry2014fast} focused on accelerating these walks using various techniques involving fast Laplacian solving, electrical flows and Schur complements. This line of research culminated in Schild's work~\cite{schild2018almost}, which achieved an almost-linear time algorithm with a runtime of $m^{1+o(1)}$. The current state-of-the-art, presented by Anari, Liu, Oveis Gharan, Vinzant and Vuong~\cite{anari2021log}, introduces a simple yet elegant approach based on ``down-up random walks'', achieving a near-optimal algorithm with a runtime $O(m \log^2 m)$. Notably, the analysis of the mixing time is both technically intricate and profound.
\end{itemize}

\subsection{Main Result}
In this work, we propose a quantum algorithm for approximately sampling random
spanning trees, leading to the theorem below.
We let $\cW_G$ denote the distribution over spanning trees of $G$, where each
tree is sampled with probability proportional to the product of its edge
weights.

\begin{restatable}[Quantum algorithm for sampling a random spanning tree]{theorem}{upperbound}\label{thm:quantum-algo-UST-informal}
  There exists a quantum algorithm $\QRST(\cO_G, \varepsilon)$ that, given query
  access $\cO_G$ to the adjacency list of a connected graph $G = (V, E, w)$
  (with $\abss{V} = n$, $\abss{E} = m$, $w \in \RR^E_{\geq 0}$), and accuracy
  parameter~$\varepsilon$, with high probability, outputs a spanning tree $T$ of
  $G$ drawn from a distribution which is $\varepsilon$-close to $\cW_G$ in total
  variation distance.
  The algorithm makes $\widetilde{O}(\sqrt{mn} \log(1/\varepsilon))$ queries
  to~$\cO_G$, and runs in $\widetilde{O}(\sqrt{mn}\log(1/\varepsilon))$ time.
\end{restatable}

We also prove a matching lower bound, showing that the runtime of our quantum algorithm is optimal up to polylog-factors.

\begin{restatable}[Quantum lower bound for sampling a random spanning tree]{theorem}{lowerbound}
\label{thm:lower-bound}
Let $\varepsilon < 1/2$ be a constant.
For any graph $G=\parens{V,E,w}$ with $\abss{V}=n$, $\abss{E}=m$, $w \in \RR^E_{\geq 0}$, consider the problem of sampling a random spanning tree from a distribution $\varepsilon$-close to $\cW_G$, given adjacency-list access to $G$. The quantum query complexity of this problem is $\Omega\parens{\sqrt{mn}}$.
\end{restatable}

\subsection{Techniques}
Our algorithm is based on the ``down-up random walk'' approach
from~\cite{anari2021log}.
Their core idea (already present in earlier work~\cite{russo2018linking}) is to
consider a random walk over spanning trees.
In each step of the random walk, they randomly remove an edge to split it into
two components, and then add a new edge sampled from the edges across the two
components proportional to the edge weights.
Their crucial contribution is to show that this random walk has a 
mixing time $\widetilde O(n)$, which is sublinear with respect to the number of edges $m$.
Since each step of the down-up random walk can be costly ($\widetilde O(m)$,
naively), their final algorithm actually considers an ``up-down random walk'',
where first an edge is added and then another edge is removed.
While this walk has a larger mixing time $\widetilde O(m)$, each step can now be
implemented in amortized time $\widetilde O(1)$ using link-cut trees.

\subsubsection{Idea (and Barrier) for Quantum Speedups}

To achieve a sublinear-time quantum algorithm, we could
follow~\cite{anari2021log} and attempt to quantumly accelerate the mixing time
of the up-down walk.
However, apart from some special
cases~\cite{aharonov2001quantum,moore2002quantum,richter2007quantum}, no general
quantum speedups for the mixing time of classical random walks are known.
As a result, any algorithm with a mixing time of $\widetilde{O}(m)$ does not
yield a quantum advantage.

A natural alternative is to revisit the down-up random walk, which has a
sublinear mixing time of $\widetilde{O}(n)$, and leverage quantum techniques to
speed up the implementation of each step.
Specifically, after removing an edge, we can sample an edge between the two
resulting components in $\widetilde{O}(\sqrt{m})$ time using Grover
search~\cite{grover1996fast}, compared to the classical $\widetilde{O}(m)$ time.
Unfortunately, this still results in an overall complexity of
$\widetilde{O}(\sqrt{m} n) \in \widetilde{\Omega}(m)$, offering no speedup over
classical algorithms.

\subsubsection{Our Approach}
To overcome the above barriers, we use an amortization idea of implementing a
``large-step'' down-up walk instead of a single-step approach.
In a standard down-up walk, a single edge in the spanning tree is changed after
each step.
In contrast, the large-step down-up walk modifies several edges---approximately
$\Theta(n)$---in each step, resulting in a significantly faster mixing time of
$\widetilde{O}(1)$.
These $\Theta(n)$ edges are sampled from a batch of $\widetilde O(n)$ edges,
which we select among the~$m$ edges in time $\widetilde O(\sqrt{m n})$ using
quantum search.
While the framework appears simple at first glance, implementing each step
requires a careful integration of various quantum algorithms.
Additionally, analyzing the mixing time relies on deep results
from~\cite{ALV22}.
Below, we outline the main techniques utilized in our approach:

\paragraph{Domain Sparsification Under Isotropy.}
Domain sparsification~\cite{anari2020isotropy,anari2021domain,ALV22} is a
framework which aims to reduce the task of sampling from a distribution on
$\binom{\sqb{m}}{k}$ with $k \ll m$ to sampling from some related distribution
on $\binom{\sqb{t}}{k}$ for $t \ll m$.
For so-called ``strongly Rayleigh distributions''~\cite{borcea2009negative}, the
domain size can be reduced to be nearly linear in the input size with
$t=\widetilde O \parens{k}$~\cite{ALV22}.
For random spanning trees, which are a canonical example of a strongly Rayleigh
distribution, this reduces the sampling problem from domain $\binom{\sqb{m}}{n}$
to $\binom{\sqb{t}}{n}$ with $t \in \widetilde O(n)$.
This provides an opportunity for quantum acceleration, provided that it is
easier to sample this subdomain of $\widetilde O(n)$ edges.

A key technique enabling efficient sampling of such a subdomain is the
``isotropic transformation'', which can be interpreted as a discrete analogue of
the isotropic position of a continuous
distribution~\cite{rudelson1999random,lee2024eldan}.
Originally developed in the context of designing algorithms for sampling from
logconcave distributions, the isotropic position uses a linear map to convert a
distribution into a standard form, enabling faster and more efficient sampling.
In~\cite{ALV22}, Anari Liu and Vuong describe a discrete isotropic
transformation for distributions over $\binom{\sqb{m}}{k}$ based on the
marginals of individual elements of $\sqb{m}$.
After this isotropic transformation, domain sparsification reduces to sampling a
subdomain uniformly at random from the domain.

For the particular case of sampling random spanning trees, the marginal
probability of an edge is given by its effective resistance, and domain
sparsification amounts to sampling edges uniformly at random in the
isotropic-transformed multigraph (see~\cref{le:isotropic-transform-on-graph}),
whose construction is based on the effective resistances of the edges in the
original graph.

In our algorithm, rather than explicitly computing this isotropic transformation
(whose description size is $\widetilde O(m)$ which we cannot tolerate), we
utilize a quantum data structure, $\QResistance$, which provides quantum
query access to effective resistances, to ``implicitly'' construct and maintain
the isotropic-transformed multigraph.
More specifically, with $\QResistance$ in place, we can efficiently
implement the up step of the walk, which involves sampling a subdomain of $O(n)$
edges uniformly from the implicit isotropic-transformed multigraph, as well as
the down step, which requires sampling a random spanning tree from the resulting
$\widetilde O(n)$-edge subdomain.

\paragraph{Quantum Graph Algorithms and Quantum Multi-Sampling.} 
To achieve the near-optimal quantum speedup, our algorithm incorporates a suite
of quantum algorithms as subroutines.
During the initialization stage, it utilizes the quantum minimum spanning tree
algorithm proposed in~\cite{DHHM06} to identify a starting spanning tree,
reducing the complexity from $\widetilde{O}(m)$ to $\widetilde{O}(\sqrt{mn})$.
For approximating the effective resistance, we adopt the approach
in~\cite{AdW22} that combines the quantum graph sparsification algorithm with
the Spielman-Srivastava toolbox~\cite{SS11} to construct an efficient
$\QResistance$ data structure in $\widetilde{O}(\sqrt{mn})$ time.
This data structure provides query access to approximate effective resistances
with a cost of $\widetilde{O}(1)$ per query.
For the up operation, we develop a variant of quantum search,
$\QIsoSample$, for sampling multiple edges in the isotropic-transformed
graph with up to a quadratic speedup.
The key idea behind this algorithm is to leverage a variant of the ``preparing
many copies of quantum states'' technique proposed in~\cite{Ham22}, utilizing
the query access provided by our $\QResistance$.
This variant of the Hamoudi's technique can be thought of as a sampling without
replacement method, ensuring that the samples are all different.
Notably, $\QIsoSample$ efficiently samples a set of $\widetilde O(n)$
edges in $\widetilde{O}(\sqrt{mn})$ time.
Leveraging these algorithms, our approach ultimately achieves random spanning
tree sampling in $\widetilde{O}(\sqrt{mn})$ time.

\subsubsection{Lower Bound}
The lower bound follows from a reduction of the problem of finding $n$ marked elements among $m$ elements, whose quantum query complexity is $\Theta(\sqrt{mn})$, to the problem of sampling a uniform spanning tree in a weighted graph.
The reduction encodes the search problem into the weights of a fixed graph, in such a way that a uniformly random spanning tree in that graph reveals the marked elements.
The argument is similar to the $\Omega(\sqrt{mn})$ lower bound from~\cite{DHHM06} on the quantum query complexity of finding a minimum spanning tree.

\subsection{Open Questions}
Our work raises several interesting questions and future directions, including the following:
\begin{itemize}
	\item The runtime of our quantum algorithm is tight, up to polylogarithmic factors, for sampling random spanning trees in weighted graphs. Can we potentially improve the runtime for unweighted graphs, for instance, to $\widetilde{O}(n)$?
    This is the case for constructing spanning trees, as was shown in the earlier work \cite{DHHM06}.
    Note that our $\Omega(\sqrt{mn})$ lower bound applies only to the weighted graph case. 
    We were unable to prove a $\omega(n)$ lower bound for the unweighted case, based on which we conjecture that there exist more efficient algorithms for the unweighted case.
	
	\item This work represents the first application of the down-up walk in quantum algorithm design. As demonstrated by a series of breakthroughs in sampling and counting problems for spin systems~\cite{alev2020improved,anari2021spectral,chen2021rapid,jain2021spectral,alimohammadi2021fractionally,liu2021coupling,blanca2022mixing,abdolazimi2022matrix}, the down-up walk has proven highly effective in designing classical algorithms for a variety of graph problems. Notable examples include independent set sampling~\cite{alev2020improved,anari2021spectral}, $q$-colorings sampling~\cite{chen2021rapid,jain2021spectral,blanca2022mixing}, edge-list-colorings sampling~\cite{abdolazimi2022matrix}, planar matching counting and sampling~\cite{alimohammadi2021fractionally}, and vertex-list-colorings sampling~\cite{liu2021coupling}. Given these successes, it would be intriguing to explore whether quantum speedups can be achieved for these problems. 
	
\item Additionally, exploring quantum speedups for determinantal point processes (DPPs) is a natural and meaningful direction. Like random spanning trees, DPPs are strongly Rayleigh distributions, enabling efficient domain sparsification~\cite{anari2020isotropy,anari2021domain,ALV22}. They also play a crucial role in applications such as machine learning~\cite{kulesza2011k,kulesza2012determinantal,derezinski2020improved}, optimization~\cite{derezinski2020debiasing,derezinski2020improved,nikolov2022proportional}, and randomized linear algebra~\cite{derezinski2017unbiased,derezinski2019minimax,derezinski2024solving}. Developing quantum algorithms for DPPs could lead to quantum computational advantages in these areas.
\end{itemize}

\section{Preliminaries}
For simplicity, 
we use $\sqb{n}$ to 
represent the set $\sets{1, 2, \ldots, n}$ and
$\sqb{n}_{0}$ to 
represent the set $\sets{0, 1, \ldots, n-1}$. 
We consider an
undirected weighted graph $G = (V, E, w)$, where $V$ is the vertex set, $E$ is
the edge set, and $w : E \mapsto \RR_{\geq0}$ represents the weight function. We use
$n,m$ to denote the size of $V$ and $E$, respectively. We use
$\delta_i$ to denote the vector whose elements are 0 except for the $i$-th
element, which is 1.

For a distribution defined over all subsets $\mu: 2 ^{\sqb{m}}\to \RR_{\geq 0}$ and $S \subseteq \sqb{m}$, let $\mu_S$ be the restricted distribution of $P \sim \mu$ conditioned on the event $P \subseteq S$.

\subsection{Quantum Computational Model}

We assume the same quantum computational model as in for
instance~\cite{DHHM06,AdW22}.
This entails a classical control system that (i) can run quantum circuits on
$O(\log n)$ qubits, (ii) can make quantum queries to the adjacency list oracle
$\cO_G$, and (iii) has access to a quantum-read/classical-write RAM of
$\widetilde O(\sqrt{mn})$ bits.
The \emph{quantum query complexity} measures the number of quantum queries that
an algorithm makes.
The \emph{quantum time complexity} measures the number of elementary operations
(classical and quantum gates), quantum queries, and single-bit QRAM operations
(classical-write or quantum-read).
If we only care about the quantum query complexity, then we can drop the QRAM
assumption at the expense of a polynomial increase in the time complexity.

\subsection{Markov Chains and Mixing time}
Let $\mu,\nu$ be two discrete probability distributions over a finite set $\Omega$.
The total variation distance (or $\TV$-distance) between $\mu$ and $\nu$ is
 defined as
  \begin{equation*}
    \norm{\mu-\nu}_{\TV}:=\frac{1}{2}\sum_{x \in \Omega}\abss{\mu\parens{x}-\nu\parens{x}}.
  \end{equation*}
The Kullback-Leibler (KL) divergence, also known as relative entropy, between $\mu$ and $\nu$
  is defined as
  \begin{equation*}
    \cD_{\KL}\parens{\mu \parallel \nu}:=\sum_{x\in \Omega}\mu\parens{x}\log\paren{\frac{\mu\parens{x}}{\nu\parens{x}}}.
 \end{equation*}
A well-known result, often summarized as ``relative entropy never increases'', is stated in the following theorem:
\begin{theorem}[Data Processing Inequality]\label{thm:data-processing}
	For any transition probability matrix $P\in \RR^{\Omega\times \Omega^\prime}$ and pair of distributions $\mu,\nu:\Omega\to \RR_{\geq 0}$, it holds that 
	\begin{equation*}
		\cD_{\KL}\parens{\nu P \parallel \mu P}\leq \cD_{\KL}\parens{\nu \parallel \mu}.
	\end{equation*}
\end{theorem}
	
A Markov chain on a finite state space $\Omega$ is specified by its transition probability matrix $P \in \RR^{\Omega \times \Omega}$. For a comprehensive introduction to the analysis of Markov chains, we refer the reader to~\cite{levin2017markov}. A distribution $\mu$ is said to be stationary with respect to $P$ if $\mu P = \mu$. In this case, $\mu$ is also called the stationary distribution of the Markov chain.
When $P$ is ergodic, the Markov chain has a unique stationary distribution $\mu$, and for any initial probability distribution $\nu$ on $\Omega$, $\norms{\nu P^t- \mu}_{\TV}$ converges to $0$ as $t \to \infty$. To precisely characterize how quickly an ergodic Markov chain converges, one can define the mixing time as below.

\begin{definition}[Mixing time]
  Let $P$ be an ergodic Markov chain on a finite state space $\Omega$ and let
  $\mu$ denote its stationary distribution.
  For any probability distribution $\nu$ on $\Omega$ and
  $\varepsilon\in \interval[open]{0}{1}$, we define
  \begin{equation*}
    t_{\mix}\parens{P,\nu,\varepsilon}:=\min\set{t \geq 0 \mid \norms{\nu P^t -\mu}_{\TV}\leq \varepsilon}.
  \end{equation*}
\end{definition}
The following well-known theorem demonstrates the relationship between the modified log-Sobolev constant and mixing times, which bounds the mixing time via modified log-Sobolev inequalities.
\begin{theorem}[\cite{bobkov2006modified}, see also corollary 8 in~\cite{CGM21}]\label{thm:MLSI-mixing}
  Let $P$ denote the transition matrix of an ergodic, reversible Markov chain on
  a finite state space $\Omega$ with the stationary distribution $\mu$.
  Suppose there exists some $\alpha\in \interval[open left]{0}{1}$ such that for
  any probability distribution $\nu$ on $\Omega$, we have
  \begin{equation*}
    \cD_{\KL}\parens{\nu P \parallel \mu P}\leq \parens{1-\alpha}\cD_{\KL}\parens{\nu\parallel \mu}.
  \end{equation*}
 Then, for any $\varepsilon\in \interval[open]{0}{1}$,
 \begin{equation*}
   t_{\mix}\parens{P,1_x,\varepsilon}\leq \ceil{\frac{1}{\alpha}\cdot\paren{\log\log\paren{\frac{1}{\mu \paren{x}}}+\log\paren{\frac{1}{2\varepsilon^2} }}},
 \end{equation*}
   where $1_x$ is the point mass distribution supported on $x$.
\end{theorem}

\subsection{Laplacian and Graph Sparsification}
For an undirected weighted graph  $G = (V, E, w)$, the weighted degree of a vertex  $v$  is defined as
$\deg(v) = \sum_{e \in E: v \in e} w_e$,
where the sum is taken over all edges  $e$  incident to  $v$, and $ w_e$  denotes the weight of edge  $e$.

\begin{definition}[Laplacian]
	The Laplacian of a weighted graph $G=\parens{V,E,w}$ is defined as the matrix $L_{G}\in \RR^{V\times V}$ such that
	\begin{equation*}
		\parens{L_ G}_{ij}=
		\begin{cases}
			\deg\parens {i} & \textup{if }i=j,             \\
			-w _ {ij}       & \textup{if }\sets{i,j}\in E, \\
			0               & \textup{otherwise}.
		\end{cases}
	\end{equation*}
\end{definition}

Equivalently, the Laplacian of graph $G$ is given by $L _G= D _G-A _G$, with $A_G$ the
weighted adjacency matrix $\paren{A_G}_{ij}=w _{ij}$ and $D$ the diagonal
weighted degree matrix $D_G=\diag\paren{\deg\paren{i}: i \in V}$. $L_G$ is a
positive semidefinite matrix since the weight function $w$ is nonnegative, and $L_G$ induces the quadratic form
\begin{equation*}\label{eq:quadratic-laplacian}
x ^\top L_ G  x=\sum_{\set {i,j}\in  E}w _{ij}\cdot \parens{x_i -x_j}^2
\end{equation*}
for  arbitrary vector $x \in \RR^V$.
Graph sparsification produces a reweighted graph with fewer edges, known as a graph spectral sparsifier.
A graph spectral sparsifier of $G$ is a re-weighted subgraph that closely approximates the quadratic form of the Laplacian for any vector $x\in \RR^V$
(see \cref{def:graph-spectral-sparsifier} for a formal definition).

In the groundbreaking work~\cite{SS11}, the authors
demonstrated that graphs can be efficiently sparsified by sampling $\widetilde O(n)$ edges with
weights roughly proportional to their effective resistance. Now, we introduce the concept of
effective resistance.

\begin{definition}[Effective Resistance]
	Given a graph $G =\parens {V,E,w}$, the \emph{effective resistance} $R_{ij}$ of a pair of $i , j \in V $ is defined as
	\begin{equation*}
		R_ {ij }=\parens {\delta_i -\delta_j }^\top L_G ^{+}\parens {\delta_i -\delta_j }= \norms {(L_G^+)^{1/2} \parens {\delta_i -\delta_j }}^2.
	\end{equation*}
	Here, $\delta_i$ is a vector with all elements equal to 0 except for the $i$-th is 1, and $L^{+}_G$ is the Moore-Penrose inverse of the Laplacian matrix $L_G$. 
\end{definition}
\begin{definition}[Leverage Score and Leverage Score Overestimates]
	Given a graph $G =\parens {V,E,w}$, the \emph{leverage score} $\ell_e$ of an edge  $e \in E$  is defined as the product of its weight and its effective resistance:
\begin{equation*}
\ell_e = w_e R_{ij}, \quad \text{for } e = \{i, j\}.
\end{equation*}
    An $\lambda$-leverage score overestimates
    $\widetilde\ell$ for $G$
    is a vector in $\RR^E_{\geq 0}$
    satisfying $\norms{\widetilde \ell}_1\leq \lambda$,
    and $\widetilde\ell_e\geq \ell_e= w_e R_e$ for all $e\in E$.
\end{definition}

\begin{lemma}[{Foster's theorem~\cite{foster1949average}}]\label{le:Foster}
  For a weighted graph $G=\parens{V,E,w}$, let $\ell_{e}$ denote the leverage score of an edge $e\in E$.
  Then, it holds that $\sum_{e \in E} \ell_e \leq n - 1$.
\end{lemma}
\begin{proof}For a edge $\sets{i,j}$, recall that
  $R_{ij}=\paren{\delta_i -\delta_j}^\top L_G^{+}\paren{\delta_i -\delta_j}$,
  then
  \begin{equation*}
  \sum_{e \in E} \ell_e =\sum_{\sets{i,j} \in E} w_{ij} R_{ij} =\sum_{\sets{{i,j}}\in E} \trace \paren{w_{ij}\paren{\delta_i -\delta_j}\paren{\delta_i -\delta_j}^\top L_G^{+}}=\trace\paren{L_G L_G^{+}}\leq n-1,
  \end{equation*}
  since the rank of $L_G$ is at most $n-1$.
\end{proof}

\subsection{Strongly Rayleigh Distributions and Random Spanning Tree}
Before defining strongly Rayleigh distributions, we first introduce the stability of generating polynomials, a fundamental property in both mathematics and physics. 
For a probability distribution $\mu:\binom{\sqb{m}}{k}\to \RR_{\geq 0}$, we define the generating polynomial $f_{\mu}$ of $\mu$ as
\begin{equation*}
	f _{\mu}\parens{z_1,\ldots, z_m}= \sum_{S \in \binom{\sqb{m}}{k} }\mu\parens{S}\cdot \prod_{i \in S}z_i .
\end{equation*}
A nonzero polynomial $f\in\CC\sqb{z_1,\ldots,z_m}$ is called stable if, $\forall i \in \sqb{m}$, the fact that $\Im\parens{z_i}>0$ implies that $f \paren{z_1,\ldots,z_m}\neq 0$. A stable polynomial with real coefficients is called real stable.

We say $\mu$ is a strongly Rayleigh distribution if and only if $f_{\mu}$ is a real stable polynomial. We say $\mu$ is $k$-homogeneous if $f_{\mu}$ is $k$-homogeneous.
A useful fact is that if $\mu$ is strongly Rayleigh, then its restricted distribution $\mu _S$ 
for arbitrary $S\subseteq \sqb{m}$ is also strongly Rayleigh. For more details, we refer the reader to~\cite{borcea2009negative}.

For any connected undirected weighted graph $G=\parens{V,E,w}$, let $\cT_G \subseteq \binom{\sqb{m}}{n-1}$ denote the set of all spanning trees of $G$. We provide the formal definition of random spanning tree below.

\begin{definition}[$w$-uniform distribution on trees]
For any connected weighted undirected graph $G =\parens{V,E,w}$, the distribution $\cW_G$ on $\cT_G$ is $w$-uniform if $\pb_{X\sim \cW_G}\sqb{ X=T}\propto \prod _{e \in T} w _e$.
\end{definition}
We refer to $\cW_G$ as the $w$-uniform distribution on $\cT_G$.
A random spanning tree drawn from the distribution $\cW_G$ is called $w$-uniform.

\begin{proposition}[Spanning Tree Marginals]\label{lemma:Spanning-Tree-Marginals}
  Given a graph $G=\parens{V,E,w}$, the probability that an edge $e=\sets{i,j} \in E $
  appears in a tree sampled $w$-uniformly randomly from $\cT_G $ is given by its
  leverage score $\ell_e = w_e R_{ij}$.
  Namely, $\pb_{ T \sim \cW_{G } }\sqb{e\in T}=\ell_e$.
\end{proposition}

\begin{proposition}[Spanning Trees are Strongly Rayleigh {\cite[Theorem 3.6]{borcea2009negative}}]
	For any connected weighted undirected graph $G =\parens{V,E,w}$, the $w$-uniform distribution $\cW_G$ is $\parens{n-1}$-homogenous strongly Rayleigh.
\end{proposition}

\begin{theorem}[Random Tree Sampling~{\cite[Theorem 2]{anari2021log}}]\label{thm:classical-random-tree-sampling}
	There is an algorithm $\RST\parens{G,\varepsilon}$ that, given a connected weighted graph $G=\parens{V,E,w}$ and $\varepsilon>0$, outputs a random spanning tree in $O\parens{m\log m\log\parens{m/\varepsilon}}$ time. The distribution of output is guaranteed to be $\varepsilon$-close to $\cW_G$ in total variation distance.
\end{theorem}

\subsection{Down-up and Up-down Walks}
For a distribution $\mu$ on $\binom{\sqb{m}}{k}$, 
we define the  complement
distribution $\bar\mu$ on $\binom{\sqb{m}}{m -k}$ by setting 
$\bar \mu \parens{S}\defeq\mu\parens{\sqb{m}\smallsetminus S}$ for every $S \in \binom{\sqb{m}}{m -k}$.

\begin{definition}[Down operator]
For $\ell \leq k$ define the row-stochastic matrix $D_{k \to \ell}\in \RR_{\geq 0}^{\binom{\sqb{m}}{k}\times \binom{\sqb{m}}{\ell}}$ by
\begin{equation*}
	D_{k \to \ell}\parens{S,T}=\begin{cases}
		0 & \text{ if } T \nsubseteq S,\\
		\frac{1 }{\binom{k}{\ell}} & \text{ otherwise. }
	\end{cases}
\end{equation*}
\end{definition}
For a distribution $\mu$ on sets of size of $k$, $\mu D_{k \to \ell}$ will be a distribution on sets of size of $\ell$. In particular, $\mu D_{k \to 1}$ will be the marginal distribution of $\mu$: $\pb_{S\sim\mu}\parens{i\in S}/k$, $i \in \sqb{m}$.
\begin{definition}[Up operator] For $\ell \leq k$ define the row-stochastic matrix $U_{\ell \to k}\in \RR_{\geq 0}^{\binom{\sqb{m}}{\ell}\times \binom{\sqb{m}}{k}}$ by
\begin{equation*}
		U_{\ell \to k}\parens{T,S}=\begin{cases}
		0 & \text{ if } T \nsubseteq S,\\
		\frac{\mu\parens{S}}{\sum_{S^\prime : T\subseteq S^\prime}\mu \parens{S^\prime} } & \text{ otherwise. }
	\end{cases}
\end{equation*}
\end{definition}
We consider the following Markov chain $M_{\mu}^t$ defined for any positive integer $t\geq k+1$, with the state space $\supp\parens{\mu}$. Starting from $S_0 \in \supp\parens{\mu}$, one step of the chain $M_{\mu}^t$ is:
\begin{enumerate}
	\item Sample $T\in \binom{\sqb{m}\smallsetminus S_0}{t -k}$  uniformly randomly.
	\item Sample $S_1\sim \mu_{S_0 \cup T}$ and update $S_0$ to be $S_1$.
\end{enumerate}

In~\cite{anari2021domain, ALV22}, it was demonstrated that the above constructed Markov chain can be used for sampling, as it possesses several desirable properties.

\begin{proposition}[Proposition 25 in~\cite{anari2021domain}, see also Proposition 15 and 16 in~\cite{ALV22}]\label{prop:complement-down-up-chain}
The complement of $S_1$ is distributed according to $\bar\mu_0 D_{\parens{m -k}\to \parens{m-t}} U _{\parens{m -t}\to \parens{m-k}}$ if we start with $S_0 \sim \mu_0$. Moreover, for any distribution $\mu: \binom{\sqb{m}}{k}\rightarrow \RR_{\geq 0}$ that is strongly Rayleigh, the chain $M_\mu^t$ for $t \geq k+1$ is irreducible, aperiodic and has stationary distribution $\mu$.
\end{proposition}

\section{Main Algorithm}

Our main quantum algorithm implements a ``large-step'' up-down walk over the spanning trees of~$G$.
Starting from a spanning tree $T$, the up-step samples $k = 2n$ edges $S$ uniformly at random from $E \setminus T$, and then the down-step samples a uniform spanning tree from the subgraph on $T \cup S$.
We can show that, if the graph is  nearly-isotropic,
meaning that its edges have approximately uniform leverage scores,
this walk achieves a mixing time of $\widetilde O(1)$ (see \cref{prop:isotropic-graph-mixing-time}).

In practice, most input graphs are not naturally isotropic. Isotropy here refers to a balanced structure where every edge contributes roughly equally to the connectivity of the graph. This balanced contribution is captured by the concept of leverage scores, which, in our setting, are given by the product of the effective resistances of the edges and their corresponding weights. The effective resistance of an edge serves as a measure of its “importance” in maintaining connectivity: edges with high effective resistance (or high leverage scores) have an outsized influence on the spanning tree distribution. If these values vary significantly across edges, the uniform sampling in the up-step may become biased, adversely affecting the mixing time of the chain.
To address this, we introduce a preprocessing step that ``isotropizes'' $G$ by approximately normalizing the effective resistances. 
Specifically, we construct a quantum data structure $\QResistance$ that enables efficient quantum queries to constant-factor approximations of the effective resistances of the edges in $G$. This builds upon the quantum graph sparsification algorithm from~\cite{AdW22}, which efficiently estimates effective resistances.
This preprocessing step, which we detail in \cref{prop:quantum-effective-resistance-oracle}, takes $\widetilde{O}(\sqrt{mn})$ time. Once initialized, it allows us to simulate quantum query access to an isotropic-transformed version of $G$, denoted as $G^{\prime}$. 
With quantum query access to $G^{\prime}$, we can then implement the up-step efficiently by uniformly sampling $k$ edges, thereby preserving the desired rapid mixing properties of the Markov chain.

As initial state of the random walk, the algorithm obtains a maximum weight-product spanning tree $T$ by running the quantum algorithm
$\QMaxProductTree$, a slightly modified version of the algorithm from~\cite{DHHM06} (see \cref{thm:min-spanning-tree-finding} for details).
We label the tree with all its edges having $1$ as a label, i.e., $e\to (e, 1)$
(the label is used for identification in the isotropic-transformed graph) and
store it in QRAM\@.
Subsequently, the algorithm executes $\widetilde{O}(1)$ rounds of the up-down
walk.
To implement the up-step, we develop a quantum algorithm $\QIsoSample$
for uniformly sampling $k$ edges $S$ in the isotropic-transformed graph.
The key idea is to combine a routine for ``preparing many copies of quantum
states'' with the leverage score oracle access given by our
$\QResistance$ algorithm.
We detail this algorithm in \cref{thm:quantum-algo-for-isotropic-sampling}.
Notably, the algorithm $\QIsoSample$ returns a sparsified domain of
$k = 2n$ edges in $\widetilde{O}(\sqrt{mn})$ time.

Using oracles to $G$ and $\cR$, we then build the subgraph with edges $T \cup S$
using the $\SubgraphConstruct$ procedure.
This subgraph has $O(n)$ edges, and so we can apply the \emph{classical}
spanning tree sampling algorithm $\RST$
(\cref{thm:classical-random-tree-sampling}) to obtain a random spanning tree in
the subgraph in $\widetilde{O}(n)$ time.
This effectively implements the down-operator.

We summarize the algorithm below.

\begin{algorithm}[H]
  \caption{Quantum Algorithm for Random Spanning Tree Sampling,
    $\QRST(\cO_G,\varepsilon)$}\label{alg:qust}
  \begin{algorithmic}[1] 
    \REQUIRE{} Quantum Oracle $\cO _G$ to a graph $G=(V,E,w)$ with
    $\abss{V}=n,\abss{E}=m$ and weights $w: E\to \RR_{\geq0}$; accuracy $\varepsilon >0$.
    \ENSURE{} A spanning tree from a distribution which is $\varepsilon$-close to $\cW_G$ in total variation distance.
    \STATE{} $\cR \gets \QResistance(\cO_G, \frac{1}{10})$.
    \STATE{} 
    $k\gets 2 n, M \gets O \parens{\log ^3 n\cdot \log (1/\varepsilon)}$.
    \STATE{} $T_{0}\gets \QMaxProductTree(\cO_{G})$,
    $T_{0}^\prime\gets \Label\parens{T_0}$.
    \STATE{}  $\cT\gets \QStoreTree(T_0^\prime)$.
    \FOR{$t=1$ to $M$}
    \STATE{} $S^\prime_{t}\gets \QIsoSample(\cO_{G}, \cT, \cR, k)$.
    \STATE{} $H_{t}^\prime\gets \SubgraphConstruct\parens{\cO_G,\cR, T^{\prime}_{t-1}, S_{t}^\prime}$.
    \STATE{} $T_t^\prime\gets \RST\parens{H^{\prime}_t, \varepsilon/\parens{2M}}$, $\cT\gets \QStoreTree(T_t^\prime)$.
    \ENDFOR{}
    \STATE{} Return the spanning tree $T^\prime_{M}$ without edge label.
  \end{algorithmic}
\end{algorithm}

More details of the subroutines employed in the algorithm are given below:
\begin{itemize}
\item
The procedure $\Label(T_0)$ means that we
add a specific label $1$ to all the edges of $T_0$. 
\item
The procedure $\QStoreTree(T)$ stores the
tree $T$ into a QRAM allowing for coherent quantum queries,
given its classical description as the input.
\item
The procedure 
$\SubgraphConstruct(\cO_G, \cR, T_{t-1}', S_t')$ constructs the subgraph after the up-step.
More precisely, given query access $\cO_G$ to 
$G = (V, E, w)$, 
query access $\cR$ to approximate resistance $\widetilde{R}$, a classical description of a tree $T_{t-1}'$
in $G$, and a set $S_t'$ of labeled edges,
we define the multigraph
$G^\prime=\parens{V,E^\prime, w^\prime}$,
where 
$E^\prime=\set{\parens{e,j} \mid e \in E, j \in \sqb{q_e} }$ with
$q_e=\ceil{\frac{mw_e\widetilde{R}_e}{n}}$ and
$w'_{\parens{e,j}} =  w_{e}/q_e$.
We treat $T_{t-1}'$
and $S_t'$
as subsets of edges in~$G^{\prime}$.
The procedure $\SubgraphConstruct$ just
outputs a classical description of the graph
$H_t' = (V, E_{t}',w_{t}')$,
where
$E_{t}' = T_{t-1}'\bigcup S_t'$,
and $w_{t}'(f) = w'(f)$ for $f\in E_{t-1}'$.
\end{itemize}

\subsection{Isotropic Transformation}


Inspired by~\cite{anari2020isotropy,anari2021domain,ALV22}, we consider a process that transforms a graph into a multigraph, ensuring that the resulting random spanning tree distribution is nearly-isotropic.

\begin{definition}[Isotropic Transformation of a Graph]\label{def:graph-isotropic-transformation}
    Consider a 
    connected weighted graph $G =\parens{V,E, w}$,
    with $\lambda$-leverage score overestimates $\widetilde\ell\in\RR^E_{\geq 0}$ for $G$.
    The isotropic transformed multigraph $G'$ of $G$
    (with respect to $\widetilde{\ell}$)
    is defined
    as $G^\prime =\parens{V,E^\prime, w ^\prime}$:
    Here $E^\prime=\set{(e,j) \mid e \in E,  j \in \sqb{q_e} }$, 
    where $(e,j)$ connects the same vertices as $e$, 
    and $q_e =\ceil{\frac{m \widetilde\ell_e}{\lambda}}$;
    $w^\prime _{(e,j)}=w _e / q_e $ for all  $e \in E$ and $ j \in \sqb{q_e}$. 
\end{definition}

\begin{lemma}[Isotropic Bound]\label{le:isotropic-transform-on-graph}
For a graph $G = (V, E, w)$
with $\abss{V} = n$ and $\abss{E} = m$,
let $\widetilde{\ell}\in \mathbb{R}^E_{\geq 0}$ be 
its $\lambda$-leverage score overestimates.
Let $G' = (V, E', w')$ denotes
the isotropic-transformed graph of $G$
with respect to $\widetilde{\ell}$.
Then,
it holds that
$\abss{E^\prime}\leq 2 \abss{E}$.
Furthermore,
for all $(e, j) \in E ^\prime$, 
we have
\[
\pb_{S \sim \cW_{G ^\prime}}\sqb{\parens{e,j}\in S} \leq\frac{\lambda}{m}.
\]
\end{lemma}
\begin{proof}
	One has
	\begin{equation*}
		\abss{E^\prime}=\sum_{e\in E}q_e\leq \sum_{e\in E}\paren{1+\frac{m \widetilde\ell_e}{\lambda}}\leq m+\frac{m\norms{\widetilde\ell}_1}{\lambda}\leq2m.
	\end{equation*}
	For any $e=\sets{a,b} \in E$ and $j\in \sqb{q_e}$, the marginal 
	\begin{equation*}
		\pb_{S \sim \cW_{G ^\prime}}\sqb{(e, j)\in S}
        = w^\prime_{\paren{e, j}} R'_{ab}
        = w^\prime_{\paren{e, j}} R_{ab}
        = \frac{w_e R_{ab}}{q_e}\leq \frac{\lambda}{m},
	\end{equation*}
    where the first equality follows from \cref{lemma:Spanning-Tree-Marginals}, and the second  equality follows from the fact that $R'_{ab} = R_{ab}$.
\end{proof}

The following theorem demonstrates that the down operator satisfies a significantly stronger entropy contraction inequality when applied to strongly Rayleigh distributions with nearly uniform marginals.

\begin{theorem}[{\cite[Theorem 28]{ALV22}}]\label{thm:isotropy-entropy-contraction}
	Let $\mu: \binom{\sqb{m}}{k}\to \RR_{\geq 0}$ be a strongly Rayleigh distribution with marginals $p\in \RR^{m}$ defined by $p_i\defeq\pb_{S \sim \mu}\sqb{ i \in S}$.
Suppose $\mu$ satisfies
\begin{equation*}
  p_{\max} \defeq \max\set{ p_i: i \in \sqb{m}} \leq  1/500.
\end{equation*}
Then for any distribution $\bar \nu:\binom{\sqb{m}}{m-k}\to \RR_{\geq 0}$ and $t \geq k +1$, we have
\begin{equation*}
	\cD_{\KL}\parens{\bar \nu D _{\parens{m-k}\to \parens{m-t}}\parallel \bar \mu D_{\parens{m-k}\to \parens{m-t}}}\leq \parens{1-\kappa} \cD_{\KL}\parens{\bar \nu\parallel \bar \mu}
\end{equation*}
with $\kappa=\Theta\paren{\frac{ t- k }{\parens{m p_{\max}+t}\log^2 m}}$. 
\end{theorem}

\begin{proposition}\label{prop:isotropic-graph-mixing-time}
  Let $G = (V, E, w)$  be a weighted graph with $\abss{V} = n$, $\abss{E} = m$, and $w \in \RR^E_{\geq 0}$. Suppose $G$ has $\lambda$-leverage score overestimates given by $\widetilde \ell\in \RR_{\geq 0}^E$, and assume $m \ge 1000 n$ and $\lambda\leq 2n$. Consider the isotropic transformed multigraph  $G'=\parens{V,E^\prime, w^\prime}$  of $G$ with respect to $\widetilde \ell$, as defined in \cref{{def:graph-isotropic-transformation}}.
    Let $t \ge n$ be an integer, $m^\prime=\abss{E^\prime}$, and define the transition matrix
    $P = D _{\parens{m^{\prime}-n+1}\to \parens{m^{\prime}-t}} U_{\parens{m^{\prime}-t}\to \parens{m^\prime-n+1}}$.
    Then, the mixing time satisfies
    \begin{equation}\label{eq:MT}
t_{\mix}\parens{P,T_{\max}^\prime,\varepsilon} = O(\log^3(n) \log(1/\varepsilon)).
    \end{equation}
Here, $T_{\max}^\prime=\sets{\parens{e,1}\mid e\in T_{\max}}\subseteq E^\prime$, where $T_{\max}$ is the maximum weight-product spanning tree of $G$. 
\end{proposition}
\begin{proof}
    For the isotropic-transformed graph $G'$,
    by \cref{le:isotropic-transform-on-graph},
    we have $m ^\prime \leq 2m$, and
    \begin{equation*}
    \pb_{S \sim \cW_{G ^\prime}}\sqb{(e, j)\in S} \le \frac{\lambda}{m} \le \frac{2n}{m} \leq  \frac{1}{500}.
    \end{equation*}
    Let $\cT_{G^\prime}$ denote the set of all spanning trees of $G^\prime$, and let $\mu^\prime:\binom{\sqb{m^\prime}}{n-1}\to \RR_{\geq 0}$ be the $w^\prime$-uniform distribution on $\cT_{G^\prime}$. Applying \cref{thm:isotropy-entropy-contraction}
    with $t = 2n$ and $k = n-1$,
    we obtain that, for any distribution $\bar{\nu^\prime}:\binom{\sqb{m^\prime}}{m ^\prime -k}\to \RR_{\geq 0}$,
    \begin{equation*}
	\cD_{\KL}\parens{\bar{\nu^\prime}D _{\parens{m^\prime-n+1}\to \parens{m^\prime-t}}\parallel \bar{\mu^\prime} D_{\parens{m^\prime-n+1}\to \parens{m^\prime-t}}}\leq \parens{1-\kappa} \cD_{\KL}\parens{\bar {\nu^\prime}\parallel \bar {\mu^\prime}},
   \end{equation*}
    where $\kappa^{-1} = O(\log^2 m^\prime) = O(\log^2 n)$.
    Thus, by
    the data processing inequality (\cref{thm:data-processing}),
    we have
\[
	\cD_{\KL}\parens{\bar {\nu^\prime} P \parallel \bar{\mu^\prime} P }\leq \parens{1-\kappa} \cD_{\KL}\parens{\bar {\nu^\prime}\parallel \bar {\mu^\prime}}.
\]
Combining the above result with the modified log-Sobolev inequality (\cref{thm:MLSI-mixing}), we conclude that 
\begin{equation}
   t_{\mix}\parens{P,T^\prime_{\max},\varepsilon}\leq \ceil{\frac{1}{\kappa}\cdot\paren{\log\log\paren{\frac{1}{\mu^\prime \parens{T^\prime_{\max}}}}+\log\paren{\frac{1}{2\varepsilon^2} }}}
\end{equation}
     Now, observe that by the pigeonhole principle we have $\mu(T_{\max}) \geq 1/\binom{m}{n-1} = \Omega(1/m^n)$. Consequently,
$\mu^\prime\parens{T^\prime_{\max}}\geq \mu\parens{T _{\max}}/(m ^\prime)^{n-1}=\Omega\parens{1/(2m)^n}
$. Thus, the claim follows as desired.
\end{proof}

As a side remark,  we note that by applying \cref{prop:complement-down-up-chain}, 
we can prove a bound of the mixing time for the complement (up-down walk) in a similar way.

\subsection{QIsotropicSample Routine}
In this subsection, we provide further details about the $\QIsoSample$ routine,
which is designed for implementing the up-step.
The goal of $\QIsoSample$ is to sample a set of $k$ distinct edges from the
isotropic-transformed multigraph uniformly at random.
Throughout the discussion, we refer to this type of task as uniformly random
$k$-subset sampling (of a set $[m]$), which means sampling a uniformly random
$k$-size subset of the set $[m]$.

A crucial distinction between uniformly sampling a $k$-size subset of the set
$[m]$ and uniformly and independently sampling $k$ times from $[m]$ lies in the
sampling method.
The former requires sampling \emph{without replacement}, while the latter
involves sampling \emph{with replacement}.
Notably, the latter can be addressed directly using the ``preparing many copies
of a quantum state'' algorithm (see
\cref{thm:prepare-many-copies-of-a-quantum-state} for more details).
Building on this, we propose a reduction from sampling \emph{with replacement}
to sampling \emph{without replacement} in the following, leading to an efficient
quantum implementation for uniformly random $k$-subset sampling
(\cref{thm:quantum-multiset-sample}) and the $\QIsoSample$ routine
(\cref{thm:quantum-algo-for-isotropic-sampling}).

We begin with the following combinatorial proposition, 
which shows that we can first perform the sampling with replacement, 
and then discard any repeated elements to accomplish 
uniformly random $k$-subset sampling.

\begin{proposition}\label{prop:sample-with-replace-for-without-replacement}
  Let $m\in \mathbb{N}$.
  Suppose $j_1, j_2, \dots, j_k \in [m]$ are integers which are uniformly and
  independently sampled from $[m]$.
  Then, conditioning on the distinctness of these integers, the distribution of
  $\{j_1, j_2, \dots, j_k\}$ is the same as the distribution of a uniformly
  random $k$-subset of $[m]$.
\end{proposition}

\begin{proof}
  Consider any fixed set $S = \{s_1, s_2, \dots, s_k\}$.
  For the uniformly random $k$-subset sampling, 
  the probability that $S$ occurs
  is exactly $\frac{1}{\binom{m}{k}}$.

  For $j_1, \dots, j_k$ 
  which are uniformly and independently 
  sampled from $[m]$, 
  the probability that they form a permutation of $\{s_1, \dots, s_k\}$
  is $k!/m^k$.
  The probability that they are pairwise distinct 
  is $m!/(m^k (m-k)!)$.
  Therefore, conditioning on the pairwise distinctness, 
  the probability that
  $\{j_1, j_2, \dots, j_k\} = S$ is just
  \[
    \frac{k!/m^k}{m!/\paren{(m-k)!m^k}} = \frac{1}{\binom{m}{k}}. \qedhere
  \]
\end{proof}

We then consider how many samples
are sufficient for $k$ distinct results,
which is formalized in the following
modified balls-into-bins lemma.
\begin{lemma}\label{lemma:balls-into-bins}
  Let $X$ be a uniformly random variable 
  taking values in $[m]$. 
  For $k\in [m]$,
  let $R$ be the set of different results from sampling $X$ for
  $\widetilde{\Theta} \rbra{k}$ times. 
  Then, with high probability,
  we have 
  $\abs{R} \ge k$.
\end{lemma}

\begin{proof}
  Let $Y_i$ denote the number of samples 
  needed for $\abs{R} = i-1$ to
  become $\abs{R} = i$. 
  Let $Y = \sum_{i=1}^{k} Y_i$.
  When $\abs{R} = i-1$, 
  the probability that a sample falls in
  $[m] - R$ is
  $ p_i = \rbra*{m-i+1}/m$.
  By definition, 
  $Y_i$ obeys the geometric distribution,
  meaning that
  \begin{equation*}
    \Pr \sbra*{Y_i = k} = {(1-p_i)}^{k-1}p_i.
  \end{equation*}
  Therefore, $\E [Y_i] = 1/p_i = m/\rbra{m-i+1}$.
  By the linearity of expectations, we have
  \begin{equation*}
    \E \sbra*{Y} 
    = \sum_{i=1}^{k} \E \sbra*{Y_i} 
    = \sum_{i=1}^{k} \frac{m}{\rbra*{m-i+1}} 
    \le \sum_{i=1}^k \frac{k}{\rbra*{k-i+1}}
    \le k\log (k).
  \end{equation*}
  Thus, by Markov's inequality, 
  with probability at least $2/3$,
  $Y\le 3k\log(k)$, 
  and the claim follows by
  amplifying the successful probability using
  Hoeffding's inequality.
\end{proof}

By combining the above propositions and 
\cref{thm:prepare-many-copies-of-a-quantum-state},
we obtain a quantum algorithm for
uniformly random $k$-subset sampling 
from a set,
which is described in the following corollary.

\begin{theorem}[Quantum uniform random $k$-subset  sampling]\label{thm:quantum-multiset-sample}
    There exists a quantum algorithm $\mathsf{MultiSample}(\cO_w, k)$ that,
    given query access $O_w$ to a vector $w\in \mathbb{N}^n$
    (where $\Abs{w}_1 = O(\textup{poly}(n))$),
    and $k\in [n]$,
    with high probability,
    returns a uniformly random $k$-size subset $S$ of
    $S_w = \{(j, \ell)|j\in [n], \ell \in [w_j]\}$.
    The algorithm uses $\widetilde{O}(\sqrt{nk})$
    queries to $\cO_w$ and runs in $\widetilde{O}(\sqrt{nk})$ time.
\end{theorem}

\begin{proof}
    In the following, we denote $k' = c k \log k$
    where $c$ is a sufficiently large constant.

    From \cref{thm:prepare-many-copies-of-a-quantum-state},
    we know with high probability, 
    the algorithm
    $\mathsf{MultiPrepare}(\cO_w, k')$ will
    return $k'$ copies of the state
    \[
        \ket{w} = \frac{1}{\sqrt{W}} \sum_{j=1}^n \sqrt{w_j}\ket{j},
    \]
    where $W = \sum_{j=1}^n w_j$,
    using $\widetilde{O}(\sqrt{nk})$ queries to $\cO_w$
    and in $\widetilde{O}(\sqrt{nk})$ time.
    We then measure the $k'$ states on the computational
    basis, obtaining $k'$ samples $j_1, j_2, \dots j_{k'}$.

    Then, for each sample $j_i$,
    we query the oracle $\cO_w$ with the state $\ket{j_i}\ket{0}$,
    and measure the state $\cO_w\ket{j_i}\ket{0}$
    in the computational basis to get the value of $w_{j_i}$. We then uniformly sample a integer $\ell_i\in [w_{j_i}]$, 
    producing a uniformly random 
    sample $(j_i, \ell_i)$ from $S_w$.
    For each sample $j_i$, 
    these operations take $\widetilde{O}(1)$ time.
    After the above procedure, we get a sequence $(j_1, \ell_1), (j_2, \ell_2), \ldots, (j_{k'}, \ell_{k'})$.

    We then output the first $k$ distinct elements 
    of the above sequence
    (if there are no such $k$ elements, the algorithm simply outputs fail; by \cref{lemma:balls-into-bins}, the probability of this situation occurring is small.).
    By \cref{prop:sample-with-replace-for-without-replacement},
    these elements constitute a set of size $k$ which
    is uniformly randomly drawn from $S_w$.
    It is clear that the above algorithm uses $\widetilde{O}(\sqrt{nk})$ queries to $\cO_w$ in total and
    runs in $\widetilde{O}(\sqrt{nk})$ time.
\end{proof}

As a direct result of \cref{thm:quantum-multiset-sample},
given adjacency list access to a graph and query access
to its approximate resistance,
we can sample $k$ different edges in
the isotropic-transformed graph in $\widetilde{O}(\sqrt{mk})$
time, 
yielding an efficient quantum implementation of the up-step.

\begin{theorem}[Quantum  Multi-Sampling in the Isotropic-Transformed Graph]\label{thm:quantum-algo-for-isotropic-sampling}
  There exists a quantum algorithm
  $\QIsoSample(\cO_G, \cT, \cR, k)$,
  that, 
  given query access $\cO_G$ to a graph $G = (V, E, W)$ 
  (where
  $\abss{V} = n$ and $\abss{E} = m$), 
  access $\cT$ to a labeled tree $T^\prime$ of the multigraph $G^\prime=\parens{V,E^\prime, w^\prime}$, access $\cR$ to an instance of $\QResistance$
  that allows quantum query to the approximate effective resistance
  $\widetilde{R}_e$ for edge $e$, and integer $k \le m$, with high probability,
  outputs a uniformly random $k$-size subset $S$ of $E^\prime\setminus T^\prime$,
  where
  $E^\prime=\set{(e,j) \mid e \in E, j \in \sqb{q_e} }$, 
  with
  $q_e=\ceil{\frac{mw_e\widetilde{R}_e}{n}}$.
  The algorithm uses $\widetilde{O}(\sqrt{mk})$ queries to $\cO_G$ and $\cR$, and runs in
  $\widetilde{O}(\sqrt{mk})$ time.
\end{theorem}

\begin{proof}
    By definition, we have
    \[
    \cO_G \ket{e}\ket{0}
     = \ket{e} \ket{w_e},
    \text{  and  }
    \cR\ket{e} \ket{0}
    = \ket{e} \ket{\widetilde{R}_e}.
    \]
    Thus,
    together with
    access $\cT$ to a labeled tree $T'$,
    we could implement a unitary $U_q$
    satisfying
    \[
        U_q\ket{e}\ket{0} = \ket{e}\ket{q_e}
    \]
    with 
    $q_e = \ceil{\frac{mw_e \widetilde{R}_e}{n}}$
    if $e\in E' \setminus E_T'$, and $q_e = 0$ if $e\in E_T'$,
    using $\widetilde{O}(1)$ queries and in $\widetilde{O}(1)$ time.
    This is because we can coherently 
    evaluate $q_e$ given $\cO_G$, $\cR$ and $\cT$.

    Then, using \cref{thm:quantum-multiset-sample}, 
    with high probability,
    the algorithm $\mathsf{MultiSample}(U_q, k)$ will return a set
    $S$ of size $k$ which is uniformly drawn from 
    $E^\prime=\set{(e,j) \mid e \in E,  j \in \sqb{q_e} }$,
    using $\widetilde{O}(\sqrt{mk})$ queries to $\cO_G$,
    and in $\widetilde{O}(\sqrt{mk})$ time.
\end{proof}

\subsection{Proof of Main Theorem}
We can now prove our main theorem.
\begin{proof}
Without loss of generality, we assume $ n/m =o \parens{1}$.
At the initial step, 
the algorithm executes the subroutine $\QResistance$
to obtain query access to $\frac{1}{10}$-overestimates effective resistances of $G$
in $\widetilde{O}\parens{\sqrt{mn}}$ time (by \cref{prop:quantum-effective-resistance-oracle}). 
It then applies $\QMaxProductTree$
(see \cref{thm:min-spanning-tree-finding}) to find a spanning tree $T_{0}$ as the start in $\widetilde{O}\parens{\sqrt{mn}}$ time. 
Each edge $e \in T_0$ is then assigned a label $e^{(1)}$, and the labeled tree is stored in $\widetilde{O}(n)$ time
in the procedure $\Label$ and $\QStoreTree$.

Subsequently, in each iteration $t \in \sqb{M}$, the algorithm first uses the
sampling subroutine $\QIsoSample$ to sample a uniformly random set
$S^\prime_{t}$ of size $k$ from $E^\prime\setminus T^\prime_{t-1}$ in
$\widetilde O \parens{\sqrt{mk}}$ time, where $k=\Theta(n)$.
Next, using~\cref{thm:classical-random-tree-sampling}, it samples a labeled
random spanning tree in the subgraph $H ^\prime_t$ in $\widetilde{O}\parens{n}$
time, since the edge-set of subgraph $T ^\prime _{t-1}\bigcup S^\prime _{t}$ has
size at most $n-1+k= O \parens{n}$.
Overall, the algorithm's time complexity is $\widetilde{O}\parens{\sqrt{mn}}$.
It remains to prove the correctness of algorithm.

We first observe that $\widetilde{R}_e$ is a $\frac{1}{10}$-overestimate of $R_e$, i.e., $R_e \leq \widetilde{R}_e\leq \frac{11}{10} R_e $ for all $e \in E$. Thus, $\widetilde{\ell}_e = w_e \widetilde{R}_e$ provides an $\frac{11}{10} n$-leverage score overestimates, as shown below:
\begin{equation*}
	\sum_{e \in E} w_e \widetilde R_e \leq \sum_{e \in E} \frac{11}{10} w_e R_e \leq \frac{11}{10} n,
\end{equation*}
where the final inequality follows from Foster’s theorem (\cref{le:Foster}).

If the sampling error for the random spanning tree in each iteration (line 8 of~\cref{alg:qust}) is zero, 
then each iteration can be viewed as one step of the chain $M_{\mu^{\prime}}^k$, 
where $\mu^{\prime}: \binom{\sqb{m^{\prime}}}{k} \to \RR_{\geq 0}$ 
corresponds exactly to the random spanning tree distribution $\cW_{G^{\prime}}$. 
Here, $m^{\prime}$ represents the number of edges 
of the isotropic multigraph $G^{\prime}$ satisfying $m^\prime\leq 2m$.
By \cref{prop:isotropic-graph-mixing-time}, the (ideal) up-down walk has mixing time $O(\log^3(n) \log(1/\varepsilon))$ (\cref{eq:MT}). 
We can hence pick $M \in O(\log^3(n) \log(1/\varepsilon))$ so as to ensure that the (ideal) up-down walk returns a tree $T'_M$ that is distributed $\varepsilon/2$-close to $\cW_{G'}$.
Since $G^{\prime}$ is obtained by equally dividing the weight of each edge in~$G$ among the multiedges, mapping the multigraph $G^{\prime}$ back to $G$ ensures that distribution of the unlabeled  tree $T_M$ remains $\varepsilon/2$-close to $\cW_G$.

Taking into account the sampling error, note that we chose the error of $\RST\parens{H^{\prime}_t, \varepsilon/\parens{2M}}$ to be $\varepsilon/(2M)$, ensuring that in each iteration, the TV-distance between the actual and ideal case is at most $\varepsilon/(2M)$. After $M$ iterations, the cumulative TV-distance is bounded by $\varepsilon/2$.
Since the ideal case returns $T_M$ that is $\frac{\varepsilon}{2}$-close to $\cW_G$, actual distribution of $T_M$ will be $\varepsilon$-close to $\cW_G$.
\end{proof}

\section{Lower Bound}

Here we prove our lower bound, showing that the algorithm's time and query complexity are optimal, up to polylogarithmic factors.

\lowerbound*
\begin{proof}
The argument is similar to the $\Omega(\sqrt{mn})$ lower bound from~\cite{DHHM06} on the quantum query complexity of finding a minimum spanning tree, and follows from a lower bound on quantum search.
Let $m = k(n+1)$ for some integer $k$.
Consider a binary matrix $M \in {\{0,1\}}^{n \times k}$, with the promise that every row of $M$ contains a single 1-entry.
The task of finding all $n$ 1-entries corresponds to solving $n$ parallel $\mathrm{OR}_k$ instances.
By the direct product theorem on the quantum query complexity of the $\mathrm{OR}_k$-function~\cite{klauck2007quantum}, the quantum query complexity of this problem is $\Theta(n \sqrt{k}) = \Theta(\sqrt{mn})$.

Now we construct a weighted graph $G_M$ from $M$, in such a way that solving the aforementioned search problem on $M$ reduces to sampling a random spanning tree from $G_M$.
The vertices of $G_M$ are $\{s,\ell_1,\dots,\ell_k,r_1,\dots,r_n\}$.
Its edges are all $(s,\ell_i)$ with edge weight 1, and all $(\ell_i,r_j)$ with edge weight $M_{ij}$.
Notably, there are $n+k+1$ vertices and $n+k$ weight-1 edges.
The union of these edges forms the unique spanning tree $T$ with nonzero weight $\Pi_{e \in T} w_e = 1$.
As a consequence, sampling a random spanning tree in $G_M$ always returns $T$.
Since $T$ identifies all weight-1 edges in $G_M$ (and hence all 1-entries in $M$), the quantum query complexity is $\Omega(\sqrt{mn})$\footnote{A more refined argument shows that the lower bound even holds when all the edge weights are nonzero.
The argument works by replacing $w_{ij} \in \{0,1\}$ by $w_{ij} \in \{1/n^4,1\}$, and noting that a random spanning tree still returns $T$ with constant probability.}.
\end{proof}

\section*{Acknowledgement}\label{sec:ack}

We would like to thank Yang P. Liu for explaining the details of their work.
Minbo Gao and Chenghua Liu would like to thank
Sebastian Zur
for helpful discussions on quantum random walks and quantum sampling algorithms. 
Minbo Gao would like to thank
Jingqin Yang for discussions
of link-cut trees used in previous spanning tree sampling algorithms.
The work is supported by National Key Research and Development Program of China
(Grant No.\ 2023YFA1009403), National Natural Science Foundation of China (Grant
No.\ 12347104), and Beijing Natural Science Foundation (Grant No.\ Z220002).
SA was supported in part by the European QuantERA project QOPT (ERA-NET Cofund 2022-25), the French PEPR integrated projects EPiQ (ANR-22-PETQ-0007) and HQI (ANR-22-PNCQ-0002), and the French ANR project QUOPS (ANR-22-CE47-0003-01).

\bibliographystyle{alpha}
\bibliography{ref}

\newcommand{\etalchar}[1]{$^{#1}$}
\begin{thebibliography}{CBGV{\etalchar{+}}13}

\bibitem[AAKV01]{aharonov2001quantum}
Dorit Aharonov, Andris Ambainis, Julia Kempe, and Umesh Vazirani.
\newblock Quantum walks on graphs.
\newblock In {\em Proceedings of the thirty-third annual ACM symposium on
  Theory of computing}, pages 50--59, 2001.

\bibitem[AASV21]{alimohammadi2021fractionally}
Yeganeh Alimohammadi, Nima Anari, Kirankumar Shiragur, and Thuy-Duong Vuong.
\newblock Fractionally log-concave and sector-stable polynomials: counting
  planar matchings and more.
\newblock In {\em Proceedings of the 53rd Annual ACM SIGACT Symposium on Theory
  of Computing}, pages 433--446, 2021.

\bibitem[AD20]{anari2020isotropy}
Nima Anari and Micha{\l} Derezi{\'n}ski.
\newblock Isotropy and log-concave polynomials: Accelerated sampling and
  high-precision counting of matroid bases.
\newblock In {\em 2020 IEEE 61st Annual Symposium on Foundations of Computer
  Science (FOCS)}, pages 1331--1344. IEEE, 2020.

\bibitem[ADVY22]{anari2021domain}
Nima Anari, Micha{\l} Derezi\'{n}ski, Thuy-Duong Vuong, and Elizabeth Yang.
\newblock {Domain Sparsification of Discrete Distributions Using Entropic
  Independence}.
\newblock In Mark Braverman, editor, {\em 13th Innovations in Theoretical
  Computer Science Conference (ITCS 2022)}, volume 215 of {\em Leibniz
  International Proceedings in Informatics (LIPIcs)}, pages 5:1--5:23,
  Dagstuhl, Germany, 2022. Schloss Dagstuhl -- Leibniz-Zentrum f{\"u}r
  Informatik.

\bibitem[AdW22]{AdW22}
Simon Apers and Ronald de~Wolf.
\newblock Quantum speedup for graph sparsification, cut approximation, and
  laplacian solving.
\newblock {\em SIAM Journal on Computing}, 51(6):1703--1742, 2022.

\bibitem[AGM{\etalchar{+}}10]{asadpour2010log}
Arash Asadpour, Michel~X. Goemans, Aleksander Mądry, Shayan~Oveis Gharan, and
  Amin Saberi.
\newblock An ${O}(\log n \log\log n)$-approximation algorithm for the
  asymmetric traveling salesman problem.
\newblock In {\em Proceedings of the 2010 Annual ACM-SIAM Symposium on Discrete
  Algorithms}, pages 379--389, 2010.

\bibitem[AL20]{alev2020improved}
Vedat~Levi Alev and Lap~Chi Lau.
\newblock Improved analysis of higher order random walks and applications.
\newblock In {\em Proceedings of the 52nd Annual ACM SIGACT Symposium on Theory
  of Computing}, pages 1198--1211, 2020.

\bibitem[Ald90]{aldous1990random}
David~J Aldous.
\newblock The random walk construction of uniform spanning trees and uniform
  labelled trees.
\newblock {\em SIAM Journal on Discrete Mathematics}, 3(4):450--465, 1990.

\bibitem[ALG{\etalchar{+}}21]{anari2021log}
Nima Anari, Kuikui Liu, Shayan~Oveis Gharan, Cynthia Vinzant, and Thuy-Duong
  Vuong.
\newblock Log-concave polynomials iv: approximate exchange, tight mixing times,
  and near-optimal sampling of forests.
\newblock In {\em Proceedings of the 53rd Annual ACM SIGACT Symposium on Theory
  of Computing}, pages 408--420, 2021.

\bibitem[ALG22]{abdolazimi2022matrix}
Dorna Abdolazimi, Kuikui Liu, and Shayan~Oveis Gharan.
\newblock A matrix trickle-down theorem on simplicial complexes and
  applications to sampling colorings.
\newblock In {\em 2021 IEEE 62nd Annual Symposium on Foundations of Computer
  Science (FOCS)}, pages 161--172. IEEE, 2022.

\bibitem[ALG24]{anari2021spectral}
Nima Anari, Kuikui Liu, and Shayan~Oveis Gharan.
\newblock Spectral independence in high-dimensional expanders and applications
  to the hardcore model.
\newblock {\em SIAM Journal on Computing}, 53(6):FOCS20--1--FOCS20--37, 2024.

\bibitem[ALV22]{ALV22}
Nima Anari, Yang~P. Liu, and Thuy-Duong Vuong.
\newblock Optimal sublinear sampling of spanning trees and determinantal point
  processes via average-case entropic independence.
\newblock In {\em 2022 IEEE 63rd Annual Symposium on Foundations of Computer
  Science (FOCS)}, pages 123--134, 2022.

\bibitem[BBL09]{borcea2009negative}
Julius Borcea, Petter Br{\"a}nd{\'e}n, and Thomas Liggett.
\newblock Negative dependence and the geometry of polynomials.
\newblock {\em Journal of the American Mathematical Society}, 22(2):521--567,
  2009.

\bibitem[BCC{\etalchar{+}}22]{blanca2022mixing}
Antonio Blanca, Pietro Caputo, Zongchen Chen, Daniel Parisi, Daniel
  Štefankovič, and Eric Vigoda.
\newblock On mixing of markov chains: Coupling, spectral independence, and
  entropy factorization.
\newblock In {\em Proceedings of the 2022 Annual ACM-SIAM Symposium on Discrete
  Algorithms (SODA)}, pages 3670--3692, 2022.

\bibitem[Bro89]{broder1989generating}
Andrei~Z Broder.
\newblock Generating random spanning trees.
\newblock In {\em FOCS}, volume~89, pages 442--447, 1989.

\bibitem[BT06]{bobkov2006modified}
Sergey~G Bobkov and Prasad Tetali.
\newblock Modified logarithmic {Sobolev} inequalities in discrete settings.
\newblock {\em Journal of Theoretical Probability}, 19:289--336, 2006.

\bibitem[CBGV{\etalchar{+}}13]{cesa2013random}
Nicolo Cesa-Bianchi, Claudio Gentile, Fabio Vitale, Giovanni Zappella, et~al.
\newblock Random spanning trees and the prediction of weighted graphs.
\newblock {\em Journal of Machine Learning Research}, 14(1):1251--1284, 2013.

\bibitem[CGM21]{CGM21}
Mary Cryan, Heng Guo, and Giorgos Mousa.
\newblock Modified log-{S}obolev inequalities for strongly log-concave
  distributions.
\newblock {\em The Annals of Probability}, 49(1):506--525, 2021.

\bibitem[CG{\v{S}}V21]{chen2021rapid}
Zongchen Chen, Andreas Galanis, Daniel {\v{S}}tefankovi{\v{c}}, and Eric
  Vigoda.
\newblock Rapid mixing for colorings via spectral independence.
\newblock In {\em Proceedings of the 2021 ACM-SIAM Symposium on Discrete
  Algorithms (SODA)}, pages 1548--1557. SIAM, 2021.

\bibitem[CMN96]{colbourn1996two}
Charles~J Colbourn, Wendy~J Myrvold, and Eugene Neufeld.
\newblock Two algorithms for unranking arborescences.
\newblock {\em Journal of Algorithms}, 20(2):268--281, 1996.

\bibitem[DBPM20]{derezinski2020debiasing}
Michal Derezinski, Burak Bartan, Mert Pilanci, and Michael~W Mahoney.
\newblock Debiasing distributed second order optimization with surrogate
  sketching and scaled regularization.
\newblock {\em Advances in Neural Information Processing Systems},
  33:6684--6695, 2020.

\bibitem[DCMW19]{derezinski2019minimax}
Micha{\l} Derezi{\'n}ski, Kenneth~L Clarkson, Michael~W Mahoney, and Manfred~K
  Warmuth.
\newblock Minimax experimental design: Bridging the gap between statistical and
  worst-case approaches to least squares regression.
\newblock In {\em Conference on Learning Theory}, pages 1050--1069. PMLR, 2019.

\bibitem[DHHM06]{DHHM06}
Christoph D\"{u}rr, Mark Heiligman, Peter H{\o}yer, and Mehdi Mhalla.
\newblock Quantum query complexity of some graph problems.
\newblock {\em SIAM Journal on Computing}, 35(6):1310--1328, 2006.

\bibitem[DKM20]{derezinski2020improved}
Michal Derezinski, Rajiv Khanna, and Michael~W Mahoney.
\newblock Improved guarantees and a multiple-descent curve for column subset
  selection and the nystrom method.
\newblock {\em Advances in Neural Information Processing Systems},
  33:4953--4964, 2020.

\bibitem[DKP{\etalchar{+}}17]{durfee2017sampling}
David Durfee, Rasmus Kyng, John Peebles, Anup~B Rao, and Sushant Sachdeva.
\newblock Sampling random spanning trees faster than matrix multiplication.
\newblock In {\em Proceedings of the 49th Annual ACM SIGACT Symposium on Theory
  of Computing}, pages 730--742, 2017.

\bibitem[DPPR17]{durfee2017determinant}
David Durfee, John Peebles, Richard Peng, and Anup~B Rao.
\newblock Determinant-preserving sparsification of {SDDM} matrices with
  applications to counting and sampling spanning trees.
\newblock In {\em 2017 IEEE 58th Annual Symposium on Foundations of Computer
  Science (FOCS)}, pages 926--937. IEEE, 2017.

\bibitem[DW17]{derezinski2017unbiased}
Michal Derezinski and Manfred~KK Warmuth.
\newblock Unbiased estimates for linear regression via volume sampling.
\newblock {\em Advances in Neural Information Processing Systems}, 30, 2017.

\bibitem[DY24]{derezinski2024solving}
Micha{\l} Derezi{\'n}ski and Jiaming Yang.
\newblock Solving dense linear systems faster than via preconditioning.
\newblock In {\em Proceedings of the 56th Annual ACM Symposium on Theory of
  Computing}, pages 1118--1129, 2024.

\bibitem[DYMZ23]{duan2023low}
Leo~L Duan, Zeyu Yuwen, George Michailidis, and Zhengwu Zhang.
\newblock Low tree-rank bayesian vector autoregression models.
\newblock {\em Journal of Machine Learning Research}, 24(286):1--35, 2023.

\bibitem[Fos49]{foster1949average}
Ronald~M Foster.
\newblock The average impedance of an electrical network.
\newblock {\em Contributions to Applied Mechanics (Reissner Anniversary
  Volume)}, 333, 1949.

\bibitem[Gro96]{grover1996fast}
Lov~K Grover.
\newblock A fast quantum mechanical algorithm for database search.
\newblock In {\em Proceedings of the twenty-eighth annual ACM symposium on
  Theory of computing}, pages 212--219, 1996.

\bibitem[GRV09]{goyal2009expanders}
Navin Goyal, Luis Rademacher, and Santosh Vempala.
\newblock Expanders via random spanning trees.
\newblock In {\em Proceedings of the twentieth annual ACM-SIAM symposium on
  Discrete algorithms}, pages 576--585. SIAM, 2009.

\bibitem[GSS11]{gharan2011randomized}
Shayan~Oveis Gharan, Amin Saberi, and Mohit Singh.
\newblock A randomized rounding approach to the traveling salesman problem.
\newblock In {\em 2011 IEEE 52nd Annual Symposium on Foundations of Computer
  Science}, pages 550--559. IEEE, 2011.

\bibitem[Gue83]{guenoche1983random}
Alain Guenoche.
\newblock Random spanning tree.
\newblock {\em Journal of Algorithms}, 4(3):214--220, 1983.

\bibitem[Ham22]{Ham22}
Yassine Hamoudi.
\newblock Preparing many copies of a quantum state in the black-box model.
\newblock {\em Physical Review A}, 105:062440, Jun 2022.

\bibitem[HPVP21]{herbster2021gang}
Mark Herbster, Stephen Pasteris, Fabio Vitale, and Massimiliano Pontil.
\newblock A gang of adversarial bandits.
\newblock {\em Advances in Neural Information Processing Systems},
  34:2265--2279, 2021.

\bibitem[HX16]{harvey2016generating}
Nicholas~JA Harvey and Keyulu Xu.
\newblock Generating random spanning trees via fast matrix multiplication.
\newblock In {\em LATIN 2016: Theoretical Informatics: 12th Latin American
  Symposium, Ensenada, Mexico, April 11-15, 2016, Proceedings 12}, pages
  522--535. Springer, 2016.

\bibitem[JPV21]{jain2021spectral}
Vishesh Jain, Huy~Tuan Pham, and Thuy~Duong Vuong.
\newblock Spectral independence, coupling with the stationary distribution, and
  the spectral gap of the glauber dynamics.
\newblock {\em arXiv preprint arXiv:2105.01201}, 2021.

\bibitem[Kir47]{kirchhoff1847ueber}
Gustav Kirchhoff.
\newblock Ueber die aufl{\"o}sung der gleichungen, auf welche man bei der
  untersuchung der linearen vertheilung galvanischer str{\"o}me gef{\"u}hrt
  wird.
\newblock {\em Annalen der Physik}, 148(12):497--508, 1847.

\bibitem[KM09]{kelner2009faster}
Jonathan~A Kelner and Aleksander Madry.
\newblock Faster generation of random spanning trees.
\newblock In {\em 2009 50th Annual IEEE Symposium on Foundations of Computer
  Science}, pages 13--21. IEEE, 2009.

\bibitem[K{\v{S}}DW07]{klauck2007quantum}
Hartmut Klauck, Robert {\v{S}}palek, and Ronald De~Wolf.
\newblock Quantum and classical strong direct product theorems and optimal
  time-space tradeoffs.
\newblock {\em SIAM Journal on Computing}, 36(5):1472--1493, 2007.

\bibitem[KT11]{kulesza2011k}
Alex Kulesza and Ben Taskar.
\newblock k-dpps: Fixed-size determinantal point processes.
\newblock In {\em Proceedings of the 28th International Conference on Machine
  Learning (ICML-11)}, pages 1193--1200, 2011.

\bibitem[KT{\etalchar{+}}12]{kulesza2012determinantal}
Alex Kulesza, Ben Taskar, et~al.
\newblock Determinantal point processes for machine learning.
\newblock {\em Foundations and Trends{\textregistered} in Machine Learning},
  5(2--3):123--286, 2012.

\bibitem[Kul90]{kulkarni1990generating}
Vidyadhar~G. Kulkarni.
\newblock Generating random combinatorial objects.
\newblock {\em Journal of Algorithms}, 11(2):185--207, 1990.

\bibitem[Liu21]{liu2021coupling}
Kuikui Liu.
\newblock {From Coupling to Spectral Independence and Blackbox Comparison with
  the Down-Up Walk}.
\newblock In Mary Wootters and Laura Sanit\`{a}, editors, {\em Approximation,
  Randomization, and Combinatorial Optimization. Algorithms and Techniques
  (APPROX/RANDOM 2021)}, volume 207 of {\em Leibniz International Proceedings
  in Informatics (LIPIcs)}, pages 32:1--32:21, Dagstuhl, Germany, 2021. Schloss
  Dagstuhl -- Leibniz-Zentrum f{\"u}r Informatik.

\bibitem[LP17a]{levin2017markov}
David~A Levin and Yuval Peres.
\newblock {\em Markov chains and mixing times}, volume 107.
\newblock American Mathematical Soc., 2017.

\bibitem[LP17b]{lyons2017probability}
Russell Lyons and Yuval Peres.
\newblock {\em Probability on trees and networks}, volume~42.
\newblock Cambridge University Press, 2017.

\bibitem[LV24]{lee2024eldan}
Yin~Tat Lee and Santosh~S Vempala.
\newblock Eldan's stochastic localization and the {KLS} conjecture:
  Isoperimetry, concentration and mixing.
\newblock {\em Annals of Mathematics}, 199(3):1043--1092, 2024.

\bibitem[MR02]{moore2002quantum}
Cristopher Moore and Alexander Russell.
\newblock Quantum walks on the hypercube.
\newblock In {\em Randomization and Approximation Techniques in Computer
  Science: 6th International Workshop, RANDOM 2002 Cambridge, MA, USA,
  September 13--15, 2002 Proceedings 5}, pages 164--178. Springer, 2002.

\bibitem[MST14]{madry2014fast}
Aleksander Madry, Damian Straszak, and Jakub Tarnawski.
\newblock Fast generation of random spanning trees and the effective resistance
  metric.
\newblock In {\em Proceedings of the twenty-sixth annual ACM-SIAM symposium on
  Discrete algorithms}, pages 2019--2036. SIAM, 2014.

\bibitem[NST22]{nikolov2022proportional}
Aleksandar Nikolov, Mohit Singh, and Uthaipon Tantipongpipat.
\newblock Proportional volume sampling and approximation algorithms for
  a-optimal design.
\newblock {\em Mathematics of Operations Research}, 47(2):847--877, 2022.

\bibitem[RBM23]{richter2023improved}
Robin Richter, Shankar Bhamidi, and Sach Mukherjee.
\newblock Improved baselines for causal structure learning on interventional
  data.
\newblock {\em Statistics and Computing}, 33(5):93, 2023.

\bibitem[Ric07]{richter2007quantum}
Peter~C Richter.
\newblock Quantum speedup of classical mixing processes.
\newblock {\em Physical Review A—Atomic, Molecular, and Optical Physics},
  76(4):042306, 2007.

\bibitem[RTF18]{russo2018linking}
Lu{\'\i}s~MS Russo, Andreia~Sofia Teixeira, and Alexandre~P Francisco.
\newblock Linking and cutting spanning trees.
\newblock {\em Algorithms}, 11(4):53, 2018.

\bibitem[Rud99]{rudelson1999random}
Mark Rudelson.
\newblock Random vectors in the isotropic position.
\newblock {\em Journal of Functional Analysis}, 164(1):60--72, 1999.

\bibitem[Sch18]{schild2018almost}
Aaron Schild.
\newblock An almost-linear time algorithm for uniform random spanning tree
  generation.
\newblock In {\em Proceedings of the 50th Annual ACM SIGACT Symposium on Theory
  of Computing}, pages 214--227, 2018.

\bibitem[SS11]{SS11}
Daniel~A. Spielman and Nikhil Srivastava.
\newblock Graph sparsification by effective resistances.
\newblock {\em SIAM Journal on Computing}, 40(6):1913--1926, 2011.

\bibitem[SSK13]{sharpnack2013detecting}
James Sharpnack, Aarti Singh, and Akshay Krishnamurthy.
\newblock Detecting activations over graphs using spanning tree wavelet bases.
\newblock In {\em Artificial intelligence and statistics}, pages 536--544.
  PMLR, 2013.

\bibitem[TAL19]{teixeira2019bayesian}
Leonardo~V Teixeira, Renato~M Assun{\c{c}}{\~a}o, and Rosangela~H Loschi.
\newblock Bayesian space-time partitioning by sampling and pruning spanning
  trees.
\newblock {\em Journal of Machine Learning Research}, 20(85):1--35, 2019.

\bibitem[Wil96]{wilson1996generating}
David~Bruce Wilson.
\newblock Generating random spanning trees more quickly than the cover time.
\newblock In {\em Proceedings of the twenty-eighth annual ACM symposium on
  Theory of computing}, pages 296--303, 1996.

\end{thebibliography}

\appendix

\section{Useful Theorems}

In this part, we revisit some useful
quantum algorithms.

We first recall the following theorems 
about quantum speedups for graph sparsifications
and approximating the effective resistances.
While we state these theorems for a general choice of $\varepsilon$, we will use these routines only with $\varepsilon = \frac{1}{10}$ in our algorithm.

\begin{definition}[Graph Spectral Sparsifier]\label{def:graph-spectral-sparsifier}
	Let $ G=\parens{V,E,w}$ be a weighted graph, and let $\tilde G =\parens{V,\tilde E, \tilde w}$ be a reweighted
	subgraph of $G$, where $\tilde w:\tilde E\rightarrow \RR_{\geq 0}$ and
	$\tilde E= \sets{e\in E: \tilde w_e>0} \subseteq E$.
	For any $\varepsilon >0$, $\tilde G $ is an $\varepsilon$-spectral sparsifier
	of $G$ if for any vector $x\in \RR^V$, the following holds:
	\begin{equation*}
		\abss{x ^\top  L_ {\tilde G}x  -x ^\top L_ { G}  x} \leq
		\varepsilon \cdot x ^\top L_ { G}  x.
	\end{equation*}
\end{definition}

\begin{theorem}[Quantum Speedups for Graph Sparsification, {\cite[Theorem 1.1]{AdW22}}]%
\label{thm:apers-spectral-sparse-appendix}
  There exists a quantum algorithm $\mathsf{GraphSparsify}(\cO _G, \varepsilon)$
  that, given query access $\cO _G$ to a weighted graph $G=\paren{V,E,w}$ and a parameter $\varepsilon\geq \sqrt{n/m}$,
  with high probability, outputs the explicit description of an
  $\varepsilon$-spectral sparsifier of $G$ with
  $\tilde O\parens {n/\varepsilon^2}$ edges. The algorithm uses
  $\tilde{O} (\sqrt{mn}/\varepsilon)$ queries to $\cO_{G}$, and runs in time
  $\tilde O\parens{ \sqrt{mn}/\varepsilon}$.
\end{theorem}

 \begin{theorem}[Approximating Effective Resistances, {\cite[Theorem 2]{SS11}}]\label{thm:approximate-resistance-compute}
     There exists a classical algorithm
     $\mathsf{ApproxResistance}(G, \varepsilon)$
      that, on input graph $G= \parens{V, E, w}$ and accuracy parameter $\varepsilon > 0$,
     with high probability,
     outputs a matrix $Z_G$ of size $p \times n $ with $p= \ceils{24\log n /\varepsilon^2}$ satisfying
     \begin{equation*}
 		\parens{1-\varepsilon}R _{ab  }\leq \norm{ Z_G \parens{\delta_a -\delta_b }}^2 \leq
 		\parens{1+\varepsilon}R_ {ab }
     \end{equation*}
     for every pair of  $a, b\in V$,
     where $R_{ab}$ is the effective resistance between $a$ and $b$.
     The algorithm runs
     in $\widetilde O\parens{m /\varepsilon^2}$ time.
 \end{theorem}

 By first applying the quantum graph sparsification algorithm,
 and then computing the approximate effective resistance on
 the sparisified graph,
 we can obtain a quantum speedup for
 approximating effective resistance overestimates.
 This is characterized by the following theorem.

\begin{proposition}[Quantum Algorithms for the Effective Resistance Overestimates, Adapted from {\cite[Claim 7.9]{AdW22}}]\label{prop:quantum-effective-resistance-oracle}
  Assume quantum query acces $\cO_G$ to $G= \parens{V, E, w}$.
  There is a quantum data structure
  $\QResistance$, that
  takes in the accuracy parameter $\varepsilon \in \interval[open]{0}{1/3} $,
  and supports the following operations:
  \begin{itemize}
    \item Initialization $\mathsf{QResistanceInit}(G, \varepsilon)$:
    outputs an instance $\mathcal{R}$
    of $\QResistance$.
    This operation uses $\tilde O\parens{\sqrt{mn}/\varepsilon}$ queries to
    $\cO_G$, and runs in
    $\tilde O\parens{\sqrt{mn}/\varepsilon+ n /\varepsilon^4}$ time.
    \item Query $\mathcal{R}.\mathsf{Query}$: outputs a unitary satisfying
          \[
          \mathcal{R}.\mathsf{Query} \ket{a} \ket{b} \ket{0} = \ket{a}\ket{b}\ket{\tilde R _{ab}}
          \]
          with
          $ R _{ab }\leq \tilde R _{ab}\leq (1+\varepsilon) R _{ab} $
          for every pair of vertices $a,b \in V$, where $R _{ab}$ is the
          effective resistance between vertices $a$ and $b$ in graph $G$.
          The operation runs in
          $\widetilde O(1/\varepsilon^2)$ time.
  \end{itemize}
\end{proposition}
\begin{proof}
	Originally the quantum algorithm in~\cite[Claim 7.9]{AdW22}
    provides query access to 
    the estimate $R_{ab}'$ of the effective resistance
    satisfying
    $(1-\varepsilon')R_{ab} \le R_{ab}' \le (1+\varepsilon') R_{ab}$.
    Therefore, 
    for our purpose,
    it suffices to take
    $\varepsilon' = \varepsilon/3$
    and 
    $\widetilde{R}_{ab} = R_{ab}'/(1-\varepsilon')$
    in their algorithm.
\end{proof}

The following theorem allows us 
to obtain a quantum speedup
for the task of preparing multiple copies of a state.

\begin{theorem}[Preparing many copies of a quantum state, {\cite[Theorem 1]{Ham22}}]\label{thm:prepare-many-copies-of-a-quantum-state}
  There exists a quantum algorithm
  $\mathsf{MultiPrepare}(\cO_w, k)$
  that, given oracle access $\cO_{w}$ to a
  vector $w \in \RR _{\geq 0}^n$ (0-indexed), and $ k\in [n]$,
  with high probability,
  outputs $k$ copies of the state $\ket{w}$, where
  $\ket{w}=\frac{1}{\sqrt{W}}\sum_{i\in \sqb{n}_0}\sqrt{w _i } \ket{i}$ with
  $W={\sum_{i\in \sqb{n}_0}w _i}$.
  The algorithm uses $\widetilde{O} (\sqrt{nk})$ queries to $\cO_{w}$, and runs
  in $\widetilde{O} (\sqrt{nk})$ time.
\end{theorem}

We further recall the quantum algorithm for 
finding a minimum spanning tree proposed in~\cite{DHHM06}.
\begin{theorem}[Quantum algorithm for finding a minimum spanning tree, Theorem 4.2 in~\cite{DHHM06}]\label{thm:min-spanning-tree-finding}
  There exists a quantum algorithm $\mathsf{QMinTree}(\cO_{G})$ that, given
  query access $\cO_{G}$ to an undirected weighted graph $G = (V, E, w)$ with
  $\abss{V} = n$ and $\abss{E}$ = m, with high probability, outputs a minimum
  spanning tree of $G$.
  The algorithm uses $O(\sqrt{mn})$ queries to $\cO_{G}$ and runs in
  $\widetilde{O} (\sqrt{mn})$ time.
\end{theorem}

In our paper, we need to find a spanning tree
with maximal weight product 
in an undirected weighted graph.
This could
be done by directly applying above algorithm
to the same graph with 
slight modifications of weight.
We give more details in the following.

\begin{corollary}[Quantum algorithm for finding a maximum spanning tree, Adapted from Theorem 4.2 in~\cite{DHHM06}]\label{cor:max-weight-spanning-tree-finding}
  There exists a quantum algorithm $\QMaxProductTree(\cO_{G})$ that, given adjacency
  query access $\cO_{G}$ to an undirected weighted graph $G = (V, E, w)$ with
  $\abss{V} = n$ and $\abss{E}$ = m, and weight $w$ bounded
  by some constant $W$, 
  with high probability, outputs a maximum weight-product
  spanning tree of $G$.
  The algorithm uses $O(\sqrt{mn})$ queries to $\cO_{G}$ and runs in
  $\widetilde{O} (\sqrt{mn})$ time.
\end{corollary}

\begin{proof}
    Note that a maximum weight-product spanning tree of 
    $G = (V, E, w)$ corresponds to a minimum spanning tree
    of the graph $G' = (V, E, w')$,
    with $w'_e = \log (W/w_e)$ for all $e\in E$.
    Given query access to $G$,
    a query access to $G'$ can be constructed in $\widetilde{O}(1)$ time.
    Therefore, by applying~\cref{thm:min-spanning-tree-finding},
    we know $\mathsf{QMinTree}(\cO_{G'})$ is what we need.
\end{proof}

}
\end{document}